\documentclass[leqno,twoside, 11pt]{amsart}

\usepackage{latexsym,esint}
\usepackage{graphicx}

\theoremstyle{plain}

\def\endproof{\hspace*{\fill}\mbox{\ \rule{.1in}{.1in}}\medskip }

\newtheorem{theorem}{Theorem}[section]
\newtheorem{corollary}[theorem]{Corollary}
\newtheorem{lemma}[theorem]{Lemma}

\theoremstyle{definition}

\newtheorem{remark}[theorem]{Remark}

\numberwithin{equation}{section}
\numberwithin{figure}{section}

\begin{document}

\title[non-Euclidean plates]
{Scaling laws for non-Euclidean plates and the 
    $W^{2,2}$ isometric immersions of Riemannian metrics}
\author{Marta Lewicka and Reza Pakzad}
\address{Marta Lewicka, University of Minnesota, Department of Mathematics, 
206 Church St. S.E., Minneapolis, MN 55455}
\address{Reza Pakzad, University of Pittsburgh, 
Department of Mathematics, 139 University Place, Pittsburgh, PA 15260}
\email{lewicka@math.umn.edu, pakzad@pitt.edu}
\subjclass{74K20, 74B20}
\keywords{non-Euclidean plates, nonlinear elasticity, 
Gamma convergence, calculus of variations}
 
\date{\today}
\begin{abstract} 
This paper concerns the elastic structures which exhibit non-zero strain at
free equilibria. Many growing tissues (leaves, 
flowers or marine invertebrates) attain complicated configurations 
during their free growth. Our study departs from 
the $3$d incompatible elasticity theory, conjectured to explain 
the mechanism for the spontaneous formation of non-Euclidean metrics.

Recall that a smooth Riemannian metric on a simply connected domain can
be realized as the pull-back metric of an orientation preserving deformation if
and only if the associated Riemann curvature tensor vanishes identically.
When this condition fails, one seeks a deformation yielding 
the closest metric realization.  
We set up a variational formulation of this problem by
introducing the non-Euclidean version of the nonlinear 
elasticity functional, and establish its $\Gamma$-convergence under the proper
scaling.  As a corollary, we obtain new necessary and sufficient conditions 
for existence of a $W^{2,2}$ isometric immersion of a given $2$d metric
into $\mathbb R^3$.
\end{abstract}

\maketitle
\tableofcontents

\section{Introduction}

Recently, there has been a growing interest in the study of flat thin sheets
which assume non-trivial configuration in the absence of exterior forces 
or impose boundary conditions. 
This phenomenon has been observed in different contexts: growing leaves,
torn plastic sheets and specifically engineered polymer gels \cite{klein}.
The study of wavy patterns along the edges of a torn plastic sheet or the
ruffled edges of leaves suggest that the sheet endeavors to reach
a non-attainable equilibrium and hence necessarily assumes 
a non-zero stress rest configuration.

In this paper, we attempt to give a possible mathematical foundation of these
phenomena, in the context of the nonlinear theory of elasticity.
The basic model, called ``three dimensional incompatible elasticity'' 
\cite{kupferman}, follows the findings of an experiment described 
in \cite{klein}.
\begin{figure}
\centering
\includegraphics[width=10cm]{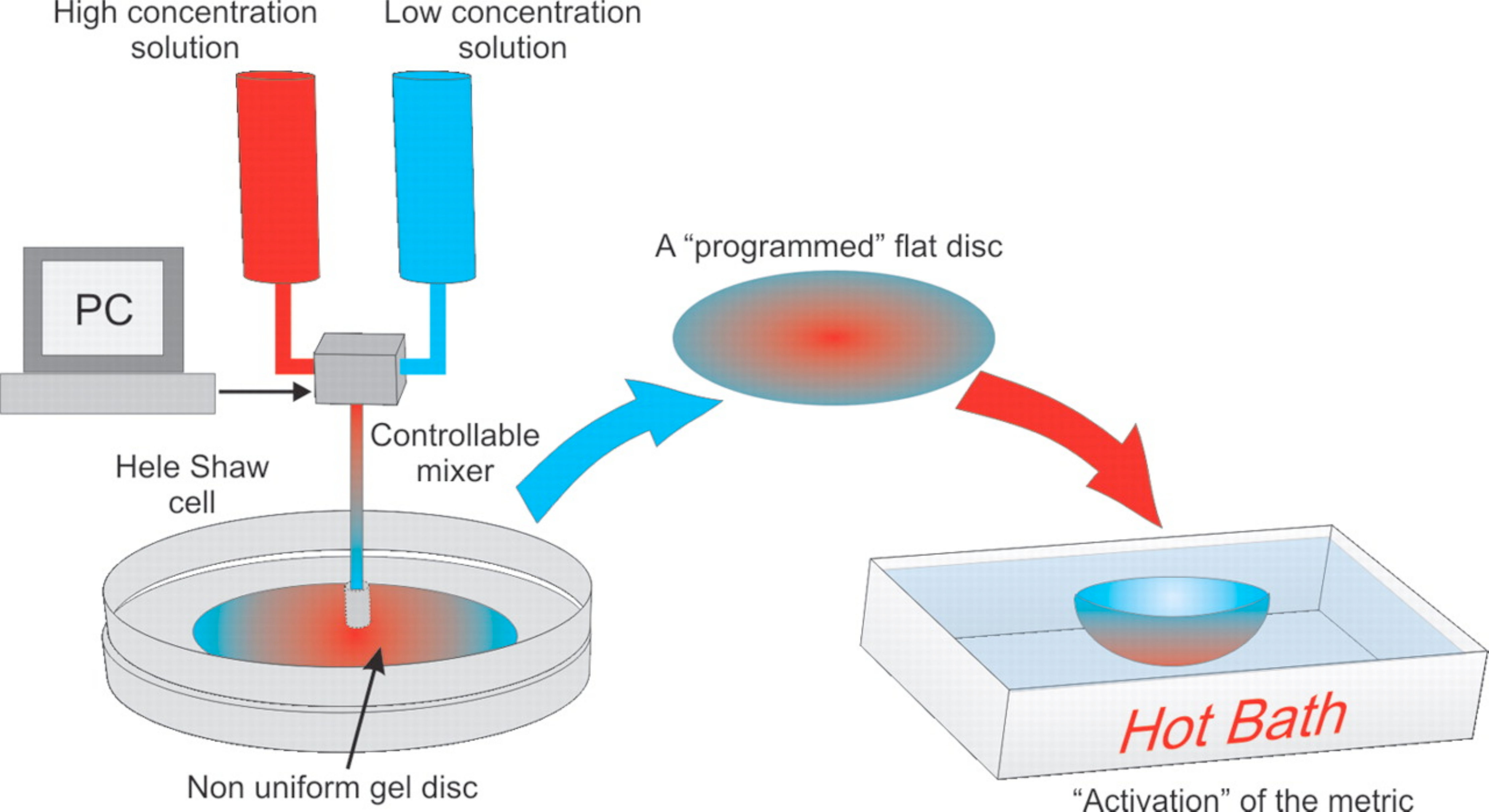}
\hspace{0.8cm}
\includegraphics[width=4.8cm]{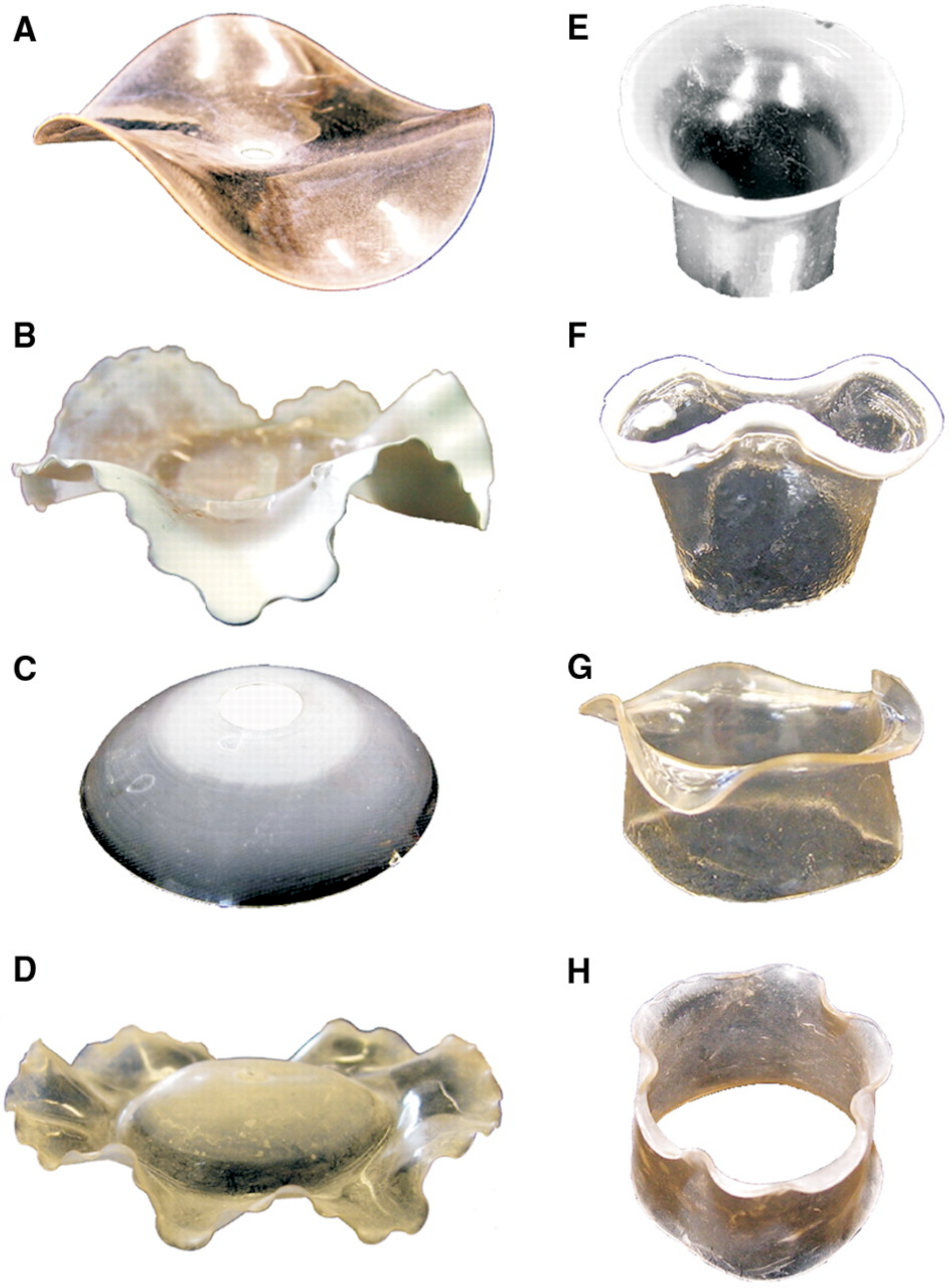}
\caption{The experimental system and the obtained structures of sheets 
with radially symmetric target metrics.
Reprinted from \cite{klein} with permission from AAAS.}
\label{figura1}
\end{figure}
The experiment (see Figure \ref{figura1}) 
consists in fabricating programmed flat disks of gels having
a non-constant monomer concentration which induces a 
``differential shrinking factor''.  
The disk is then activated in a temperature raised above
a critical threshold, whereas the gel shrinks with a factor proportional 
to its concentration and the distances between the points on the surface 
are changed. This defines a new target metric on the disk, inducing hence 
a $3$d configuration in the initially 
planar plate. One of the most remarkable features of this deformation 
is the onset of some ``transversal'' oscillations (wavy patterns). 

Trying to understand the above phenomena in the context of nonlinear 
elasticity theory, it has been postulated \cite{klein, kupferman} that
the $3$d elastic body seeks to realize a configuration with 
a prescribed pull-back metric $g$.
In this line, a $3$d energy functional was introduced in \cite{kupferman},
measuring the $L^2$ distance of the realized pull-back metric of
the given deformation from $g$. 

Unfortunately, this functional is not suitable from a variational point
of view, due to the existence of Lipschitz deformations with zero energy level
\cite{gromov}. These deformations are however neither
orientation preserving nor reversing in any neighborhood of a point
where the Riemann curvature of the metric $g$ does not vanish
(i.e. when the metric is non-Euclidean). 
In order to overcome this shortcoming,
here we introduce a modified energy $I(u)$ which measures,
in an average pointwise manner, how far a given deformation $u$ is from being
an orientation preserving realization
of the prescribed metric. An immediate consequence is that for
non-Euclidean $g$, the infimum of $I$ (in absence of any forces or boundary
conditions) is strictly positive, which points to the existence of
non-zero strain at free equilibria.

Several interesting questions arise in the study of the proposed energy
functional. A first one is to determine the scaling of the infimum energy
in terms of the vanishing thickness of a sheet.
Another is to determine limiting zero-thickness
theories under obtained scaling laws. The natural
analytical tool in this regard is that of $\Gamma$ convergence,
in the context of Calculus of Variations.

In this paper, we consider a first case where the prescribed metric  
is given by a tangential Riemannian metric $[g_{\alpha\beta}]$
on the $2$d mid-plate, and is independent of the thickness variable.
The $3$d metric $g$ is set-up such that no stretching may happen
in the direction normal to the sheet in order to realize the metric.
Consequently, if $[g_{\alpha\beta}]$
has non-zero Gaussian curvature, then such $g$ is non-Euclidean.
We further observe a correspondence between
the scaling law for the infimum energy of the thin sheet in terms
of the thickness, and the immersability of  $[g_{\alpha\beta}]$ into $\mathbb R^3$
(Theorem \ref{isoim}).
This result relates to a  longstanding problem in differential geometry, 
depending heavily on the regularity of the immersion \cite{kuiper, HH}.
In our context, we deal with $W^{2,2}$ immersions not studied previously.
We also derive the $\Gamma$-limit of the rescaled energies,
expressed by a curvature functional
on the space of all $W^{2,2}$ realizations of $[g_{\alpha\beta}]$ in 
$\mathbb R^3$ (Theorems \ref{liminf} and \ref{threcseq}).

To put our results in another context, recall the seminal work of Friesecke,
James and M\"uller \cite{FJMgeo}, where the nonlinear bending theory
of plates (due to Kirchhoff) was derived as the $\Gamma$-limit of
the classic theory of nonlinear elasticity, under the assumption that
the later energy per unit thickness $h$ scales like $h^2$.
From a mathematical point of view, the present paper provides the non-Euclidean
version of the same results under the same scaling law, and
the $2$d limit theory we obtain is the natural
non-Euclidean generalization of the Kirchhoff model.
Contrary to the classic case, in our context the scaling law is the unique 
natural scaling of the energy for
the free thin sheet with the associated prescribed metric.

As a major ingredient of proofs, we also give a generalization of the geometric 
rigidity estimate \cite{FJMgeo} to the non-Euclidean setting (Theorem
\ref{rigidity_main}). We estimate the average $L^2$ oscillations 
of the deformation gradient from a fixed matrix in terms of the $3$d 
non-Euclidean energy $I$ and certain geometric parameters of the $3$d domain.
The main difference is an extra term of the bound, depending on
the derivatives of the prescribed metric $g$ and hence vanishing when $g$
is Euclidean as in \cite{FJMgeo}.

\bigskip

\noindent{\bf Acknowledgments.}
We are grateful to Stefan M\"uller for a significant shortening of 
the original proof of Theorem \ref{rigidity_main}.
The subject of non-Euclidean plates has been brought to our attention
by Raz Kupferman.
M.L. was partially supported by the NSF grant DMS-0707275
and by the Center for Nonlinear Analysis (CNA) under
the NSF grants 0405343 and 0635983.
R.P. was partially supported by the University of Pittsburgh grant
CRDF-9003034.

\section{Overview of the main results}

Consider an open, bounded, Lipschitz domain
$\mathcal{U}\subset \mathbb{R}^n$, with a given smooth Riemannian 
metric $g = [g_{ij}]$.
The matrix field $g:\bar{\mathcal{U}}\longrightarrow\mathbb{R}^{n\times n}$
is therefore symmetric and strictly positive definite up to the boundary
$\partial\mathcal{U}$.
Let $A=\sqrt{g}$ be the unique symmetric positive definite square root of $g$
and define, for all $x\in\bar{\mathcal U}$:
\begin{equation}\label{setF}
\mathcal{F}(x) = \Big\{RA(x); ~ R\in SO(n)\Big\},
\end{equation}
where $SO(n)$ stands for the special orthogonal group of rotations
in $\mathbb{R}^n$. By polar decomposition theorem, it easily follows that
$u$ is an orientation preserving realization of $g$:
$$(\nabla u)^T\nabla u = g \quad \mbox{ and } \quad \det \nabla u>0
\qquad \mbox{ a.e. in } \mathcal{U}$$
if and only if:
$$\nabla u(x)\in\mathcal{F}(x) \qquad \mbox{a.e. in } \mathcal{U}.$$
Motivated by this observation, we define:
\begin{equation}\label{Icorrect}
I(u)= \int_{\mathcal{U}} \mbox{dist}^2(\nabla u(x), \mathcal{F}(x))~\mbox{d}x
\qquad \forall u\in W^{1,2}(\mathcal{U},\mathbb{R}^n).
\end{equation}
Notice that when $g=\mathrm{Id}$ then the above functional becomes
$I(u) = \int \mbox{dist}^2(\nabla u, SO(n))$ which is a standard quadratic
nonlinear elasticity energy, obeying the frame invariance.

\begin{remark}
For a deformation $u:\mathcal{U}\longrightarrow \mathbb{R}^n$
one could define the energy as the difference between
its pull-back metric on $\mathcal{U}$ and $g$:
\begin{equation*}
I_{str}(u)= \int_{\mathcal{U}} |(\nabla u)^T\nabla u - g|^2 ~\mbox{d}x.
\end{equation*}
However, such ``stretching'' functional is not appropriate
from the variational point of view, for the following reason.
It is known that there always exists
$u\in W^{1,\infty}(\mathcal{U},\mathbb{R}^n)$ such that $I_{str}(u) = 0$.
On the other hand \cite{gromov}, if the Riemann curvature tensor $R$ associated
to $g$ does not vanish identically, say $R_{ijkl}(x)\neq 0$ for some
$x\in\mathcal{U}$, then
$u$ must have a 'folding structure' around $x$; the realization $u$ cannot
be orientation preserving (or reversing) in any open neighborhood of $x$.
\end{remark}

In view of the above remark, our first observation concerns the energy
(\ref{Icorrect}) in case of $R\not\equiv 0$.

\begin{theorem}\label{thbound_below}
If the Riemann curvature tensor $R_{ijkl} \not\equiv 0$, then
$\inf\Big\{I(u); ~u\in W^{1,2}(\mathcal{U},\mathbb{R}^n)\Big\}>0$.
\end{theorem}

In case $g= \mbox{Id}$, the infimum as above is naturally $0$ and is attained only by
the rigid motions. In \cite{FJMgeo}, the authors
proved an optimal estimate of the deviation  in $W^{1,2}$ of a deformation $u$
from rigid motions in terms of $I(u)$. In section \ref{secrigidity} 
we give a generalization of such quantitative
rigidity estimate to our setting. Since there is no realization of $I(u) =0$ 
if the Riemann curvature is non-zero, we choose to estimate the
deviation of the deformation from a linear map at the expense of an
extra term, proportional to the gradient of $g$.

\begin{theorem}\label{rigidity_main}
For every $u\in W^{1,2}(\mathcal{U},\mathbb{R}^n)$ there exists
$Q\in\mathbb{R}^{n\times n}$ such that:
$$\int_{\mathcal{U}}|\nabla u(x) - Q|^2~\mathrm{d}x\leq
C \left(\int_{\mathcal{U}} \mathrm{dist}^2(\nabla u,
\mathcal{F}(x))~\mathrm{d}x
+ \|\nabla g\|_{L^\infty}^2 (\mathrm{diam } ~\mathcal{U})^2 
|\mathcal{U}|\right),$$
where the constant $C$ depends on $\|g\|_{L^\infty}$, $\|g^{-1}\|_{L^\infty}$,
and on the domain $\mathcal{U}$. The dependence on $\mathcal{U}$ is uniform for a family
of domains which are bilipschitz equivalent with controlled Lipschitz constants.
\end{theorem}
 
We shall consider a class of more general $3$d non-Euclidean elasticity
functionals:
\begin{equation*}
I_W(u) =  \int_{\mathcal{U}} W(x,\nabla u(x))~\mbox{d}x,
\end{equation*}
where the inhomogeneous stored energy density
$W:\mathcal{U}\times \mathbb{R}^{n\times n}\longrightarrow \mathbb{R}_+$
satisfies the following assumptions of frame
invariance, normalization, growth and regularity:
\begin{itemize}
\item[(i)] $W(x, RF) = W(x,F)$ for all $R\in SO(n)$,
\item[(ii)] $W(x, A(x)) = 0$,
\item[(iii)] $W(x,F) \geq c \mbox{ dist}^2(F,\mathcal{F}(x))$, with some
uniform constant $c>0$,
\item[(iv)] $W$ has regularity $\mathcal{C}^2$ in some neighborhood of the set
$\{(x,F); ~ x\in\mathcal{U}, F\in \mathcal{F}(x)\}$.
\end{itemize}
The properties (i) -- (iii) are assumed to hold for all $x\in \mathcal{U}$
and all $F\in \mathbb{R}^{n\times n}$. In case the Riemann curvature tensor of
$g$ does not vanish, by Theorem \ref{thbound_below} the infimum of $I_W$
is positive, in which case $I_W$ is called a three dimensional
incompatible elasticity functional.


\medskip

We consider thin $3$d plates of the form:
$$\Omega^h = \Omega\times (-h/2, h/2)\subset \mathbb{R}^3, \qquad 0<h< <1,$$
with a given mid-plate  $\Omega$ an open bounded subset of $\mathbb{R}^2$.
In accordance with \cite{klein}, we assume that the metric $g$ on
$\Omega^h$ has the form:
\begin{equation}\label{metryka}
g(x', x_3) = \left[\begin{array}{cc} \Big[g_{\alpha\beta}(x')\Big] &
\begin{array}{c} 0 \\ 0\end{array}\\
\begin{array}{cc} 0 & 0 \end{array} & 1
\end{array}\right] \qquad \forall x'\in\Omega, \quad x_3\in (-h/2, h/2),
\end{equation}
where $[g_{\alpha\beta}]$ is a smooth metric on $\Omega$, defined up to
the boundary.  In particular, $g$ does not depend on the thin variable $x_3$.
Accordingly, we shall assume that the energy density $W$ does not depend on $x_3$:
\begin{itemize}
\item[(v)] $W(x,F) = W(x',F)$, for all $x\in\mathcal{U}$ and 
all $F\in\mathbb{R}^{n\times n}$.
\end{itemize}

\medskip

Define now the rescaled energy functionals:
$$I^h(u) = \frac{1}{h}\int_{\Omega^h} W(x, \nabla u(x))~\mbox{d}x 
\qquad \forall u\in W^{1,2}(\Omega^h,\mathbb{R}^3),$$
where the energy well $\mathcal{F}(x) = \mathcal{F}(x') = SO(3) A(x)$
is given through the unique positive definite square root $A=\sqrt{g}$
of the form:
\begin{equation*}
A(x', x_3) = \left[\begin{array}{cc} \Big[A_{\alpha\beta}(x')\Big] &
\begin{array}{c} 0 \\ 0\end{array}\\
\begin{array}{cc} 0 & 0 \end{array} & 1
\end{array}\right] \qquad \forall x'\in\Omega, \quad x_3\in (-h/2, h/2).
\end{equation*}
By an easy direct calculation, one notices that the Riemann curvature tensor
$R_{ijkl} \equiv 0$, of $g$ in $\Omega^h$ if and only if the Gaussian curvature
of the $2$d metric $\kappa_{[g_{\alpha\beta}]}\equiv 0$. 
Hence, by Theorem \ref{thbound_below}, 
$\inf I^h>0$ for all $h$ if this condition is violated.
A natural question is now to investigate the behavior of the sequence  $\inf
I^h$ as $h\to 0$.
We first obtain (in section \ref{secliminf}) the following lower bound
and compactness result:

\begin{theorem}\label{liminf}
Assume that a given sequence of deformations 
$u^h\in W^{1,2}(\Omega^h,\mathbb{R}^3)$ satisfies:
\begin{equation}\label{h2}
I^h(u^h) \leq Ch^2,
\end{equation}
where $C>0$ is a uniform constant. Then, for some sequence of constants
$c^h\in \mathbb{R}^3$, the following holds for the renormalized
deformations $y^h(x', x_3) = u^h(x',hx_3) - c^h\in
W^{1,2}(\Omega^1,\mathbb{R}^3)$:
\begin{itemize}
\item[(i)] $y^h$ converge, up to a subsequence, in $W^{1,2}(\Omega^1,\mathbb{R}^3)$ to
$y(x',x_3) = y(x')$ and $y\in W^{2,2}(\Omega,  \mathbb{R}^3)$.
\item[(ii)] The matrix field $Q(x')$ with columns
$Q(x') = \Big[\partial_1 y(x'), \partial_2 y(x'), \vec n(x')\Big]\in
\mathcal{F}(x')$, for a.e. $x'\in \Omega$.
Here:
\begin{equation}\label{normal}
\vec n = \frac{\partial_1 y\times\partial_2 y}
{|\partial_1 y\times\partial_2y|}
\end{equation}
is the (well defined) normal to the image surface $y(\Omega)$.
Consequently, $y$ realizes the mid-plate
metric: $(\nabla y)^T\nabla y = [g_{\alpha\beta}]$.
\item[(iii)] Define the following quadratic forms:
$$\mathcal{Q}_3(x')(F) = \nabla^2 W (x', \cdot)_{\mid A(x')}(F,F),
\quad \mathcal{Q}_2(x')(F_{tan})
= \min\{\mathcal{Q}_3(x')(\tilde F); ~~ \tilde F_{tan}
= F_{tan} \}.$$
Then we have the lower bound:
$$\liminf_{h\to 0} \frac{1}{h^2} I^h(u^h)\geq  \frac{1}{24}\int_\Omega
\mathcal{Q}_2(x') \Big(A_{\alpha\beta}^{-1} (\nabla y)^T
\nabla\vec n\Big)~\mathrm{d}x'.$$
\end{itemize}
\end{theorem}

We further prove that the lower bound in (iii) above is optimal, in the 
following sense. 
Let $y\in W^{2,2}(\Omega,\mathbb{R}^3)$ be a Sobolev regular isometric
immersion of the given mid-plate metric, that is 
$(\nabla y)^T\nabla y = [g_{\alpha\beta}]$.
The normal vector $\vec n\in W^{1,2}(\Omega, \mathbb{R}^3)$ is then given by
(\ref{normal}) and it is well defined because
$|\partial_1 y\times\partial_2 y| = (\det g)^{1/2}>0$.

\begin{theorem}\label{threcseq}
For every isometric immersion $y\in W^{2,2}(\Omega,\mathbb{R}^3)$ of $g$,
there exists a sequence of ``recovery'' deformations
$u^h\in W^{1,2}(\Omega^h,\mathbb{R}^3)$ such that the assertion (i) of
Theorem \ref{liminf} hold, together with:
\begin{equation}\label{limexact}
\lim_{h\to 0} \frac{1}{h^2} I^h(u^h) = \frac{1}{24}\int_\Omega
\mathcal{Q}_2(x') \Big(A_{\alpha\beta}^{-1} (\nabla y)^T
\nabla\vec n\Big)~\mathrm{d}x'.
\end{equation}
\end{theorem}

A corollary of Theorems \ref{liminf} and \ref{threcseq},
proved in section \ref{isometrie}, provides a necessary
and sufficient condition for the existence of $W^{2,2}$ isometric immersions
of $(\Omega, [g_{\alpha\beta}])$:

\begin{theorem}\label{isoim}
Let $[g_{\alpha\beta}]$ be a smooth metric on the midplate
$\Omega\subset\mathbb{R}^2$. Then:
\begin{itemize}
\item[(i)] $[g_{\alpha\beta}]$ has an isometric immersion
$y\in W^{2,2}(\Omega,\mathbb{R}^3)$ if and only if $\frac{1}{h^2}\inf I^h\leq
C$, for a uniform constant $C$.
\item[(ii)] $[g_{\alpha\beta}]$ has an isometric immersion
$y\in W^{2,2}(\Omega,\mathbb{R}^2)$ (or, equivalently,
the Gaussian curvature $\kappa_{[g_{\alpha\beta}]}\equiv 0$)
if and only if $\lim_{h\to 0}\frac{1}{h^2}\inf I^h = 0$.
\item[(iii)] If the Gaussian curvature $\kappa_{[g_{\alpha\beta}]}\not\equiv 0$
in $\Omega$ then $\frac{1}{h^2}\inf I^h \geq c> 0$.
\end{itemize}
\end{theorem}

The existence (or lack of thereof) of local or global isometric
immersions of a given $2$d Riemannian manifold into $\mathbb{R}^3$
is a longstanding problem in differential
geometry, its main feature being finding the optimal regularity.
By a classical result of Kuiper \cite{kuiper},
a $\mathcal{C}^1$ isometric embedding into $\mathbb R^3$  can be 
obtained by means of convex integration (see also \cite{gromov}). 
This regularity is far from $W^{2,2}$, where information about the second
derivatives is also available.
On the other hand, a smooth isometry exists for some special cases, e.g. for
smooth metrics with uniformly positive or negative Gaussian curvatures on bounded
domains in ${\mathbb R^2}$ (see \cite{HH}, Theorems 9.0.1 and 10.0.2).
Counterexamples to such theories are largely unexplored.
By \cite{iaia}, there exists an analytic metric $[g_{\alpha\beta}]$ with
nonnegative Gaussian curvature on $2$d sphere,
with no $\mathcal{C}^3$ isometric embedding. However such metric
always admits a $\mathcal{C}^{1,1}$ embedding (see \cite{GL} and \cite{HZ}).
For a related example see also \cite{pogo}.

\medskip

Finally, notice that Theorems \ref{liminf} and \ref{threcseq} can be summarized
using the language of $\Gamma$-convergence \cite{dalmaso}.
Recall  that a sequence of functionals
$\mathcal{F}^h:X\longrightarrow \overline{\mathbb{R}}$ defined on a metric
space $X$, is said to $\Gamma$-converge, as $h\to 0$, to
$\mathcal{F}:X\longrightarrow \overline{\mathbb{R}}$ provided that
the following two conditions hold:
\begin{itemize}
\item[(i)] For any converging sequence $\{x^h\}$ in $X$:
\begin{equation*}\label{Gamma1}
\mathcal{F}\left(\lim_{h\to 0} x^h\right) \leq \liminf_{h\to 0}
\mathcal{F}^h(x^h).
\end{equation*}
\item[(ii)] For every $x\in X$, there exists a sequence $\{x^h\}$ converging
to $x$ and such that:
\begin{equation*}\label{Gamma2}
\mathcal{F}(x) = \lim_{h\to 0} \mathcal{F}^h(x^h).
\end{equation*}
\end{itemize}

\begin{corollary}
The sequence of functionals $\mathcal{F}^h: W^{1,2}(\Omega^{1},\mathbb{R}^3)
\longrightarrow \overline{\mathbb{R}}$, given by:
$$\mathcal{F}^h(y(x)) = \frac{1}{h^2}I^h(y(x',hx_3))$$
$\Gamma$-converges, as $h\to 0$, to:
\begin{equation*}
\mathcal{F}(y) = \left\{\begin{array}{ll}
{\displaystyle \frac{1}{24}\int_\Omega
\mathcal{Q}_2(x') \Big(A_{\alpha\beta}^{-1} (\nabla y)^T
\nabla\vec n\Big)~\mathrm{d}x'}
& \mbox{ if $y$ is a $W^{2,2}$ isometric immersion on $[g_{\alpha\beta}]$}
\\
+\infty & \mbox{ otherwise.}
\end{array}\right.
\end{equation*}
Consequently, the (global) approximate minimizers of $\mathcal{F}^h$
converge to a global minimizer of $\mathcal{F}$.
\end{corollary}

\section{The non-Euclidean elasticity functional - a proof of Theorem 
\ref{thbound_below}}

In the sequel, we will need some differential geometry notation.
We shall write $|g| = \det g$ and $g^{-1}=[g^{ij}]$.
The Christoffel symbols are then given, using the Einstein summation, as:
$$\Gamma^m_{ij} = \frac{1}{2} g^{km}(\partial_i g_{kj} + \partial_j g_{ik}
-\partial_kg_{ij}).$$
By $\nabla$ we denote the covariant gradient of a scalar/vector field or 
a differential form, while by $\nabla_g$ we denote the contravariant gradient. 
The covariant divergence of a vector field $u$ can be written as: 
$$\mbox{div}_g u = (\nabla_i u)^i = 
\frac{1}{\sqrt{|g|}} \partial_i (\sqrt{|g|} u^i)$$
and the scalar product of two  vector fields (that is of two 
$(1,0)$ contravariant tensors) has the form: 
$\langle u,v \rangle_g= u^i g_{ij} v^j$.
We shall often use
the Laplace-Beltrami operator $\Delta_g$ of scalar fields $f$:
$$\Delta_g f = \mbox{div}_g (\nabla_g f) =
\frac{1}{\sqrt{|g|}} \partial_i (\sqrt{|g|} g^{ij}\partial_j f).$$
By $R=[R_{ijkl}]$ we mean the ($(0,4)$ covariant)
Riemann curvature tensor, and by $Ric_g$ the ($(0,2)$ covariant) Ricci curvature 
tensor.

Towards the proof of Theorem \ref{thbound_below},
we first derive an auxiliary result, which is somewhat standard in
differential geometry (see e.g. \cite{reshetnyak, ciarlet_new}). 

\begin{lemma}\label{thsmoothandRicci}
Let $u\in W^{1,1}(\mathcal{U},\mathbb{R}^n)$ satisfy 
$\nabla u(x)\in \mathcal{F}(x)$ for a.a. $x\in\mathcal{U}$.  Then $u$ is smooth
and $R\equiv 0$.
\end{lemma}

\begin{proof}
Write $u=(u^1,\ldots, u^n)$ and notice that in view of the assumption, 
each $u^i\in W^{1,\infty}$. Moreover:
$$\det \nabla u = \sqrt{|g|}, \qquad 
\mbox{cof } \nabla u = \sqrt{|g|} (\nabla u) g^{-1}.$$
Recall that for a matrix $F\in\mathbb{R}^{n\times n}$, $\mbox{cof } F$ 
denotes the matrix of cofactors of $F$, that is 
$(\mbox{cof } F)_{ij} = (-1)^{i+j} \det \hat F_{ij}$, where 
$\hat F_{ij}\in\mathbb{R}^{(n-1)\times (n-1)}$ is obtained from $F$ by deleting 
its $i$th row and $j$th column.
Since $\mbox{div}(\mbox{cof } \nabla u) = 0$ (the divergence 
of the cofactor matrix is always taken row-wise), 
the Laplace-Beltrami operator of each component $u^m$ is zero:
$$\Delta_g u^m = 0,$$
and therefore we conclude that $u^m\in \mathcal{C}^\infty$. The second statement 
follows immediately since $u : \mathcal{U} \to \mathbb{R}^n$ is a smooth isometric embedding of  
$(\mathcal{U}, g)$ into the Euclidean space $\mathbb{R}^n$.
\end{proof}

\begin{remark}
For the convenience of the reader we now give a simple argument proving that
the existence of a smooth $u$ as in Lemma \ref{thsmoothandRicci} 
implies that the Ricci curvature
tensor $Ric_g\equiv 0$. Recall that when $n=3$ (which is the dimension
relevant to our main results), $Ric_g\equiv 0$ if and only if
$R\equiv 0$ \cite{peterson}.

We shall first deduce that the second $g$-covariant derivative of each scalar field 
$u^m$ vanishes. Since the vectors $\{\partial_s u\}_{s=1}^n$ 
form a basis of $\mathbb{R}^n$ 
($\nabla u$ being invertible), it is enough to consider the following linear 
combination of components of $\nabla^2u^m$ (for fixed covariant indices $i,j$): 
\begin{equation*}
\begin{split}
\sum_{m=1}^n \nabla_i(\nabla u^m)_j \cdot \partial_s u^m
&= (\partial^2_{ij}u - \Gamma_{ij}^k\partial_ku)\partial_su\\
&= \partial^2_{ij}u\partial_su - \frac{1}{2} g^{kl} (\partial_j g_{il}
+ \partial_i g_{jl} - \partial_l g_{ij})\partial_k u \partial_s u\\
&= \partial^2_{ij}u\partial_su - \frac{1}{2} g^{kl} (\partial_j g_{is}
+ \partial_i g_{js} - \partial_s g_{ij}) = 0.
\end{split}
\end{equation*}
Hence $\nabla^2u^m=0$ and since $|\nabla^2u^m|_g^2 = |\nabla^2_gu^m|^2_g$,
we also see that the second contravariant gradient of each component $u^m$
vanishes:
\begin{equation}\label{glupie}
\nabla^2_gu^m = 0.
\end{equation}
On the other hand, $(\nabla u)^T(\nabla u) g^{-1} = \mbox{Id}$, so
$(\nabla u) g^{-1} (\nabla u)^T = \mbox{Id}$ which means precisely that
$\partial_iu^m g^{ij}\partial_j u^m =1$. Therefore:
\begin{equation}\label{glupie2}
|\nabla_gu^m|_g^2 = |\nabla u^m|^2_g = 1.
\end{equation}
Applying now (\ref{glupie}) and (\ref{glupie2}) in the following
Bochner's formula (see e.g. \cite{peterson}):
\begin{equation*}\label{bochner}
\frac{1}{2}\Delta_g |\nabla_g f|^2 = 
\langle \nabla_g\Delta_g f, \nabla_g f\rangle_g + |\nabla_g^2f|^2 
+ Ric_g (\nabla_g f, \nabla_g f),
\end{equation*}
where we take the scalar field $f=u^m$, we obtain:
$$ Ric_g (\nabla_gu^m, \nabla_gu^m) = 0.$$
Since $\{\nabla_gu^m\}_{m=1}^n$ form a basis of $\mathbb{R}^n$ and $Ric_g$ is
a symmetric bilinear form, we conclude that $Ric_g\equiv 0$.
\endproof
\end{remark}

We now prove two further auxiliary results.

\begin{lemma}\label{lem2.3}
There is a constant $M>0$, depending only on $\|g\|_{L^\infty}$ and such that
for every $u\in W^{1,2}(\mathcal{U},\mathbb{R}^n)$ there exists a truncation
$\bar u\in W^{1,2}(\mathcal{U},\mathbb{R}^n)$ with the properties:
$$ \|\nabla\bar u\|_{L^\infty} \leq M, \qquad  
\|\nabla u - \nabla\bar u\|_{L^2(\mathcal{U})}^2\leq 4 I(u)
\qquad \mathrm{and} \qquad I(\bar u)\leq 10 I(u).$$
\end{lemma}
\begin{proof}
Use the approximation  result of Proposition A.1.
in \cite{FJMgeo} to obtain the truncation $\bar u = u^\lambda$, for $\lambda>0$
having the property that if a matrix $F\in \mathbb{R}^{n\times n}$ satisfies
$|F|\geq\lambda$ then:
$$|F|\leq 2\mbox{dist}^2(F,\mathcal{F}(x)) \qquad \forall x\in\mathcal{U}.$$
Then $\|\nabla u^\lambda\|_{L^\infty}\leq C\lambda:= M$ and further:
$$\|\nabla u - \nabla u^\lambda\|_{L^2}^2 \leq \int_{\{|\nabla u|>\lambda\}}
|\nabla u|^2 \leq 4 \int_{\{|\nabla u|>\lambda\}} 
\mbox{dist}^2(\nabla u, \mathcal{F}(x))~\mbox{d}x \leq 4 I(u).$$
The last inequality of the lemma follows from the above by triangle inequality.
\end{proof}

\begin{lemma}\label{lem2.4}
Let $u\in W^{1,\infty}(\mathcal{U},\mathbb{R}^n)$ and define vector field $w$
whose each component $w^m$ satisfies:
$$\Delta_g w^m = 0 \quad \mbox{ in } \mathcal{U}, \qquad
w^m = u^m \quad \mbox{ on } \partial\mathcal{U}.$$
Then $\|\nabla (u-w)\|_{L^2(\mathcal{U})}^2\leq C I(u)$, where the constant
$C$ depends only on the coercivity constant of $g$ and (in a nondecreasing
manner) on $\|\nabla u\|_{L^\infty}$.
\end{lemma}
\begin{proof}
The unique solvability of the elliptic problem in the statement follows 
by the usual Lax-Milgram and compactness arguments. 
Further, the correction $z=u-w\in W^{1,2}_0(\mathcal{U},\mathbb{R}^n)$ satisfies:
$$ \int_\mathcal{U} g^{ij}\sqrt{|g|} \partial_i z^m \partial_j\phi
= \int_{\mathcal{U}} g^{ij}\sqrt{|g|} \partial_i u^m\partial_j\phi
- \int_\mathcal{U} \nabla\phi (\mbox{cof } \nabla u)_{m-\mbox{th row}}$$
for all $\phi\in W^{1,2}_0(\mathcal U)$.
Indeed, the last term above 
equals to $0$, since the row-wise divergence of the cofactor
matrix of $\nabla u$ is $0$, in view of $u$ being Lipschitz continuous.
Use now $\phi=z^m$ to obtain:
\begin{equation}\label{glupie3}
\begin{split}
\int_{\mathcal{U}}\sqrt{|g|} |\nabla z^m|_g^2 & =
\int_{\mathcal{U}}\sqrt{|g|} |\nabla_g z^m|_g^2\\
&= \int_{\mathcal{U}} \partial_j z^m \left(g^{ij}\sqrt{|g|}\partial_i u^m 
- (\mbox{cof } \nabla u)_{mj}\right)\\
&\leq \|\nabla z^m\|_{L^2(\mathcal{U})} \left(\int_{\mathcal{U}}\left|\sqrt{|g|} (\nabla u)
g^{-1} - \mbox{cof }\nabla u \right|^2\right)^{1/2}\\
&\leq C_M \|\nabla z^m\|_{L^2(\mathcal{U})} I(u)^{1/2},
\end{split}
\end{equation}
which proves the lemma.

In order to deduce the last bound in (\ref{glupie3}), consider the function 
$f(F) = \sqrt{|g|} F g^{-1} - \mbox{cof } F$, which is locally Lipschitz
continuous, uniformly in $x\in\mathcal{U}$. 
Clearly, when $F\in \mathcal{F}(x)$ then $F=RA$ for some $R\in SO(n)$, and so:
$\mbox{cof } F =  \mbox{cof } (RA) = (\det A) R A^{-1} = \sqrt{|g|} (RA) g^{-1}$,
implying: $f(F) = 0$. Hence:
$$|f(\nabla u(x))|^2\leq C_M^2 \mbox{dist}^2(\nabla u(x), \mathcal{F}(x)),$$
where $C_M$ stands for the Lipschitz constant of $f$ 
on a sufficiently large ball, 
whose radius is determined by the bound $M = \|\nabla u\|_{L^\infty}$.
\end{proof}

\medskip

\noindent {\bf Proof of Theorem \ref{thbound_below}. }
\noindent We argue by contradiction, assuming that 
for some sequence of deformations
$u_n\in W^{1,2}(\mathcal{U}, \mathbb{R}^n)$, there holds 
$\lim_{n\to\infty} I(u_n) = 0$.
By Lemma \ref{lem2.3}, replacing $u_n$ by $\bar u_n$, 
we may also and without loss of generality have
$u_n\in W^{1,\infty}(\mathcal{U}, \mathbb{R}^n)$ and 
$\|\nabla u_n\|_{L^\infty}\leq M$.

Clearly, the uniform boundedness of $\nabla u_n$ implies, 
via the Poincar\'{e} inequality, after a 
modification by a constant and after passing to a subsequence if necessary:
\begin{equation}\label{glupie2.5}
\lim u_n = u \qquad \mbox{ weakly in } W^{1,2}(\mathcal{U}).
\end{equation}
Consider the splitting $u_n=w_n+z_n$ as in Lemma \ref{lem2.4}.
By the Poincar\'e inequality, Lemma \ref{lem2.4} implies 
that the sequence $z_n\in W_0^{1,2}(\mathcal{U})$ converges to $0$:
\begin{equation*}
\lim z_n = 0 \qquad \mbox{ strongly in } W^{1,2}(\mathcal{U}).
\end{equation*}
In view of the convergence in (\ref{glupie2.5}),  the sequence $w_n$ must
be uniformly bounded in $W^{1,2}(\mathcal{U})$, and hence by the local elliptic
estimates for the Laplace-Beltrami operator, each $\Delta_g$-harmonic component 
$w_n^m$ is locally uniformly bounded in a higher Sobolev norm:
$$\forall \mathcal{U}'\subset\subset \mathcal{U} \quad \exists C_{\mathcal{U}'} 
\qquad \|w_n^m\|_{W^{2,2}(\mathcal{U}')} \leq C_{\mathcal{U}'} 
\|w_n^m\|_{W^{1,2}(\mathcal{U})}\leq C.$$
Consequently, $w_n$ converge to $u$ strongly in $W_{loc}^{1,2}(\mathcal{U})$
and recalling that $I(u_n)$ converge to $0$, we finally obtain:
$$ I(u) = 0.$$
Therefore $\nabla u\in\mathcal{F}(x)$ for a.a. $x\in\mathcal{U}$,
which achieves the desired contradiction with the assumption 
$R\not\equiv 0$, by Lemma \ref{thsmoothandRicci}.
\endproof

\section{A geometric rigidity estimate for Riemannian metrics - a proof
of Theorem \ref{rigidity_main}}\label{secrigidity}

Recall that according to the basic rigidity estimate \cite{FJMgeo}, for every
$v\in W^{1,2}(\mathcal{V},\mathbb{R}^n)$ defined on an open, bounded set 
$\mathcal{V}\subset\mathbb{R}^n$, there exists $R\in SO(n)$ such that
\begin{equation}\label{basrig} 
\int_{\mathcal{V}} |\nabla v - R|^2 \leq C_{\mathcal{V}} 
\int_{\mathcal{V}} \mbox{dist}^2(\nabla v, SO(n)). 
\end{equation}
The constant $C_{\mathcal{V}}$ depends only on the domain $\mathcal{V}$
and it is uniform for a family of domains which are bilipschitz equivalent
with controlled Lipschitz constants.

\medskip

\noindent {\bf A proof of Theorem \ref{rigidity_main}.}
For some $x_0\in \mathcal{U}$ denote $A_0=A(x_0)$ and apply (\ref{basrig}) 
to the vector field $v(y) = u(A^{-1}_0y)\in W^{1,2}(A_0\mathcal{U}, \mathbb{R}^n)$.
After change of variables we obtain:
$$ \exists R\in SO(n) \qquad
\int_\mathcal{U} |(\nabla u) A^{-1}_0 - R|^2 \leq C_{A_0\mathcal{U}} 
\int_{\mathcal{U}} \mbox{dist}^2((\nabla u) A_0^{-1}, SO(n)). $$
Since the set $A_0\mathcal{U}$ is a bilipschitz image of $\mathcal{U}$, the constant
$C_{A_0\mathcal{U}}$ has a uniform bound $C$ 
depending on $\|A_0\|$, $\|A_0^{-1}\|$ and $\mathcal{U}$.  Further:
\begin{equation*}
\begin{split}
\int_\mathcal{U} |\nabla u - RA_0|^2 &\leq C\|A_0\|^4  
\int_{\mathcal{U}} \mbox{dist}^2(\nabla u, SO(n) A_0)\\
&\leq  C \|A_0\|^4 \left(
\int_{\mathcal{U}} \mbox{dist}^2(\nabla u, \mathcal{F}(x))
+ \int_{\mathcal{U}} |A(x) - A_0|^2\right)\\
&\leq  C \|g\|_{L^\infty}^2 
\left(\int_{\mathcal{U}} \mbox{dist}^2(\nabla u, \mathcal{F}(x))~\mbox{d}x
+ C \|\nabla g\|_{L^\infty}^2 (\mbox{diam } \mathcal{U})^2 |\mathcal{U}|\right),
\end{split}
\end{equation*}
which proves the claim.
\endproof

\medskip

We now derive a crucial approximation result, as in
Theorem 10 \cite{FJMhier} (see also Lemma 8.1 \cite{lemopa1}).

\begin{lemma}\label{lemapprox}
There exist matrix fields $Q^h\in W^{1,2}(\Omega,\mathbb{R}^{3\times 3})$ such that:
\begin{equation}\label{pr2} 
\frac{1}{h}\int_{\Omega^h}|\nabla u^h(x) - Q^h(x')|^2~\mathrm{d}x
\leq C (h^2 + I^h(u^h)),
\end{equation}
\begin{equation}\label{pr3} 
\int_\Omega |\nabla Q^h|^2 \leq C(1 + h^{-2}I^h(u^h)),
\end{equation}
with constant $C$ independent of $h$.
\end{lemma}
\begin{proof} 
Let $D_{x',h}=B(x',h)\cap\Omega$ be $2$d curvilinear discs
in $\Omega$ of radius $h$ and centered at a given $x'$. 
Let $B_{x',h}=D_{x',h}\times(-h/2, h/2)$ be the corresponding $3$d cylinders.
On each $B_{x',h}$ use Theorem \ref{rigidity_main} to obtain:
\begin{equation}\label{pr1}
\begin{split}
\int_{B_{x',h}} |\nabla u^h - Q_{x',h}|^2 &\leq C \left(\int_{B_{x',h}}\mbox{dist}^2
(\nabla u^h,\mathcal{F}(z))~\mbox{d}z 
+ h^2 |B_{x',h}|\right)\\
&\leq C \int_{B_{x',h}} h^2 + \mbox{dist}^2 (\nabla u^h,\mathcal{F}(z))~\mbox{d}z,
\end{split}
\end{equation}
with a universal constant $C$ in the right hand side above, depending only on 
the metric $g$ and the Lipschitz constant of $\partial\Omega$.

Consider now the family of mollifiers 
$\eta_{x'}:\Omega\longrightarrow \mathbb{R}$,
parametrized by $x'\in\Omega$ and given by:
$$\eta_{x'}(z') 
= \frac{\theta(|z' - x'|/h)}{h\int_{\Omega}\theta(|y'-x'|/h)~\mbox{d}y'},$$
where $\theta\in\mathcal{C}_c^\infty([0,1))$ is a nonnegative cut-off function, 
equal to a nonzero constant in a neighborhood of $0$. Then $\eta_{x'}(z')=0$
for all $z'\not\in D_{x,h}$ and:
$$\int_\Omega \eta_{x'} = h^{-1},\quad
\|\eta_{x'}\|_{L^\infty}\leq Ch^{-3}, \quad 
\|\nabla_{x'}\eta_{x'}\|_{L^\infty}\leq Ch^{-4}.$$
Define the approximation $Q^h\in W^{1,2}(\Omega,\mathbb{R}^{3\times 3})$:
$$Q^h(x')= \int_{\Omega^h} \eta_{x'}(z') \nabla u^h(z)~\mbox{d}z.$$
By (\ref{pr1}) , we obtain the following pointwise estimates, for every $x'\in\Omega$:
\begin{equation*}
\begin{split}
|Q^h(x') - Q_{x',h}|^2 &\leq \left( \int_{\Omega^h}\eta_{x'}(z') 
\left(\nabla u^h(z) - Q_{x',h}\right)~\mbox{d}z\right)^2 \\
&\leq \int_{\Omega^h}|\eta_{x'}(z')|^2~\mbox{d}z \cdot
\int_{B_{x',h}} |\nabla u^h - Q_{x',h}|^2
\leq Ch^{-3} \int_{B_{x',h}} h^2 + \mbox{dist}^2(\nabla u^h,\mathcal{F}(z))
~\mbox{d}z,
\end{split}
\end{equation*}
\begin{equation*}
\begin{split}
|\nabla Q^h(x')|^2 & = \left(\int_{\Omega^h}(\nabla_{x'}\eta_{x'}(z'))
\nabla u^h(z)~\mbox{d}z\right)^2\\
& = \left(\int_{\Omega^h}(\nabla_{x'}\eta_{x'}(z'))
\left(\nabla u^h(z) - Q_{x',h}\right)~\mbox{d}z\right)^2
\leq \int_{\Omega^h}|\nabla_{x'}\eta_{x'}(z')|^2~\mbox{d}z \cdot
\int_{\Omega^h} |\nabla u^h - Q_{x',h}|^2\\
&\leq C h^{-5} \int_{B_{x',h}} h^2 + \mbox{dist}^2(\nabla u^h,\mathcal{F}(z))
~\mbox{d}z.
\end{split}
\end{equation*}
Applying the same estimate on doubled balls $B_{x', 2h}$ we arrive at:
\begin{equation*}
\begin{split}
\int_{B_{x',h}}|\nabla u^h(x) - Q^h(x')|^2~\mbox{d}x &\leq 
C \left( \int_{B_{x',h}}|\nabla u^h(z) - Q_{x',h}|^2~\mbox{d}z
+ \int_{B_{x',h}}|Q_{x',h} - Q^h(z')|^2~\mbox{d}z\right)^2\\
&\leq C \int_{B_{x',2h}} h^2 + \mbox{dist}^2(\nabla u^h,\mathcal{F}(z))~\mbox{d}z,
\end{split}
\end{equation*}
\begin{equation*}
\int_{D_{x',h}}|\nabla Q^h|^2 \leq 
Ch^{-3} \int_{B_{x',2h}} h^2 
+ \mbox{dist}^2(\nabla u^h,\mathcal{F}(z))~\mbox{d}z.
\end{equation*}
Consider a finite covering $\Omega = \bigcup D_{x',h}$  
whose intersection number is independent of $h$ (as it depends only on the 
Lipschitz constant of $\partial\Omega$). 
Summing the above bounds and applying the uniform lower bound 
$W(x,F) \geq c \mbox{ dist}^2(F,\mathcal{F}(x))$ directly yields 
(\ref{pr2}) and (\ref{pr3}).
\end{proof}

\section{Compactness and the lower bound on rescaled energies
- a proof of Theorem \ref{liminf}}\label{secliminf}

{\bf 1.} From (\ref{pr2}), (\ref{pr3}) in Lemma \ref{lemapprox}
and the assumption on the energy scaling, 
it follows directly that the sequence $Q^h$, obtained in Lemma \ref{lemapprox},
is bounded in $W^{1,2}(\Omega,\mathbb{R}^{3\times 3})$. 
Hence, $Q^h$ converges weakly in this space, to some matrix field $Q$ and:
\begin{equation*}
\int_{\Omega^1}|\nabla u^h(x',hx_3) - Q(x')|^2~\mbox{d}x
\leq \int_{\Omega}|Q^h - Q|^2 + h^{-1}\int_{\Omega^h}|\nabla u^h(x) - Q^h(x')|
~\mbox{d}x,
\end{equation*}
converges to $0$ by (\ref{pr2}).
Therefore we obtain the following convergence of the matrix field with 
given columns:
$$ \lim_{h\to 0}\Big[\partial_1 y^h(x), \partial_2 y^h(x),
h^{-1}\partial_3 y^h(x)\Big] = Q(x) \qquad \mbox{ in } 
L^2(\Omega^1,\mathbb{R}^{3\times 3}).$$
We immediately conclude that $\|\partial_3 y^h\|_{L^2(\Omega^1)}$ converges to $0$.

Now, setting $c^h = \fint_{\Omega^1} u^h(x',hx_3)~\mbox{d}x$, by means of 
the Poincar\'e inequality there follows the assertion (i) of the Theorem.
The higher regularity of $y$ can be deduced from $\nabla y \in W^{1,2}$,
in view of the established $W^{1,2}$
regularity of the limiting approximant $Q$.

To prove (ii), notice that by (\ref{pr2}),  (\ref{pr1}) and the lower bound on $W$:
\begin{equation}\label{pr4}
\int_\Omega\mbox{dist}^2(Q^h, \mathcal{F}(x')) \leq Ch^2.
\end{equation}
Hence $Q(x')\in\mathcal{F}(x)$ a.e. in $\Omega$ and consequently 
$\partial_\alpha y \cdot \partial_\beta y = g_{\alpha\beta}$. 
To see that the last column of the matrix field $Q$ coincides with the unit 
normal to the image surface: $Qe_3=\vec n$, write $Q=RA$ for some $R\in SO(3)$
and notice that:
$$\partial_1 y\times \partial_2 y = (RAe_1)\times (RAe_2)
= R((Ae_1)\times (Ae_2)) = c RAe_3 = cQe_3, $$
where $c= |\partial_1y \times\partial_2 y| = |(Ae_1)\times (Ae_2)|
= \det A >0$, by the form of the matrix $A$. 
On the other hand $|Qe_3| = |Re_3|=1$, so indeed there must be $Qe_3=\vec n
=  (\partial_1y \times\partial_2 y) / |\partial_1y \times\partial_2 y| $.

\medskip

{\bf 2.} We now modify the sequence $Q^h$ to retain its convergence properties
and additionally get $\tilde Q^h(x')\in\mathcal{F}(x)$ for a.a. $x\in \Omega$.
Define $\tilde Q^h\in L^2(\Omega,\mathbb{R}^3)$ with:
$$ \tilde Q^h(x') = \left\{\begin{array}{ll}
\pi_{\mathcal{F}(x)}(Q^h(x'))  & \mbox{ if } Q^h(x')\in 
\mbox{small ngbhood of } \mathcal{F}(x)\\
A(x) & \mbox{otherwise}\end{array}
\right.$$
where $\pi_{\mathcal{F}(x)}$ denotes the projection onto the compact set 
$\mathcal{F}(x)$ of its (sufficiently small) neighborhood.
One can easily see that:
$$\int_{\Omega}|\tilde Q^h - Q^h|^2\leq 
C\int_\Omega\mbox{dist}^2(Q^h(x'),\mathcal{F}(x'))~\mbox{d}x'\leq Ch^2$$
by (\ref{pr4}).  In particular, $\tilde Q^h$ converge to $Q$
in $L^2(\Omega)$.

Write $\tilde Q^h = R^hA$ for a matrix field $R^h\in SO(3)$ and consider
the rescaled strain:
$$ G^h (x',x_3) = \frac{1}{h} \Big((R^h)^T(x')\nabla u^h(x', hx_3) - A(x')\Big)
\in L^2(\Omega^1,\mathbb{R}^{3\times 3}).$$
We obtain:
$$\int_{\Omega^1}|G^h|^2 \leq Ch^{-3} \int_{\Omega^h} |\nabla u^h - \tilde Q^h|^2
\leq Ch^{-2}I^h(u^h) + Ch^{-2}\int_{\Omega}|\tilde Q^h - Q^h|^2 \leq C.$$
Hence there exists a limit:
\begin{equation}\label{pr4.5}
\lim_{h\to 0} G^h = G \qquad 
\mbox{ weakly in } L^2(\Omega^1,\mathbb{R}^{3\times 3}).
\end{equation}

\medskip

{\bf 3.} Fix now a small $s>0$ and consider the difference quotients:
$$f^{s,h}(x) = \frac{1}{h}\frac{1}{s} (y^h(x+se_3) - y^h(x)) 
\in W^{1,2}(\Omega^1,\mathbb{R}^3).$$
Since $h^{-1}\partial_3 y^h$ converges in $L^2(\Omega^1,\mathbb{R}^3)$
to $\vec n(x')$, then also:
$$\lim_{h\to 0}f^{s,h}(x) = \lim_{h\to 0}\frac{1}{h}\fint_0^s
\partial_3 y^h(x+te_3)~\mbox{d}t = \vec n(x').$$
There also follows convergence of normal derivatives, strongly in $L^1(\Omega^1)$:
$$\lim_{h\to 0}\partial_3 f^{s,h}(x) 
= \lim_{h\to 0} (\partial_3 y^h(x+se_3) - \partial_3y^h(x)) = 0,$$
and of tangential derivatives, weakly in $L^2(\Omega^1)$:
$$\lim_{h\to 0} \partial_\alpha f^{s,h}(x) =
\lim_{h\to 0} \frac{1}{s} R^h(x') \Big(G^h(x+se_3) - G^h(x)\Big)e_\alpha=
\frac{1}{s} (QA^{-1})(x') \Big(G(x+se_3) - G(x)\Big)e_\alpha, $$
where we have used that the $L^\infty$ sequence $R^h=\tilde Q^h A^{-1}$
converges in $L^2(\Omega)$ to $QA^{-1}\in SO(3)$.

Consequently, the sequence $f^{s,h}$ converges, as $h\to 0$ weakly in
$L^2(\Omega^1)$ to $\vec n(x')$. Equating the derivatives, we obtain:
$$\partial_\alpha \vec n(x') = \frac{1}{s} (QA^{-1})(x') 
\Big(G(x+se_3) - G(x)\Big)e_\alpha.$$
Therefore:
\begin{equation}\label{pr5}
G(x',x_3)e_\alpha = G(x',0) e_\alpha + x_3 \Big((AQ^{-1})(x') 
\partial_\alpha\vec n(x')\Big), \qquad \alpha=1,2.
\end{equation}

\medskip

{\bf 4.} We now calculate the lower bound of the rescaled energies.
To this end, define the sequence of characteristic functions:
$${\displaystyle \chi_h = \chi_{\{x\in\Omega^1; ~~|G^{h}(x)|\leq h^{-1/2}\}}},$$
which by (\ref{pr4.5}) converge in $L^1(\Omega^1)$ to $1$.
Using frame invariance and noting that in the Taylor expansion of the function 
$F\mapsto W(x,F)$ at $A(x)$ the first two terms are
$0$, we obtain:
\begin{equation*}
\begin{split}
\frac{1}{h^2} I^h(u^h) &\geq \frac{1}{h^2} \int_{\Omega^1} \chi_h(x)
W(x, \nabla u^h(x',hx_3))~\mbox{d}x\\
& = \frac{1}{h^2} \int_{\Omega^1} \chi_h(x) W (x, R^h(x')^T
\nabla u^h(x',hx_3))~\mbox{d}x\\
& = \frac{1}{h^2} \int_{\Omega^1} \chi_h(x) W (x, A(x)
+ hG^h(x))~\mbox{d}x\\
& \geq \int_{\Omega^1} \chi_h(x) \left[\frac{1}{2} \nabla^2 W
(x, \cdot)_{\mid A(x)} (G^h(x), G^h(x)) - 
{o}(1) |G^h(x)|^2\right]~\mbox{d}x.
\end{split}
\end{equation*}
Hence:
\begin{equation*}
\begin{split}
\liminf_{h\to 0}\frac{1}{h^2} I^h(u^h) &\geq \frac{1}{2} \liminf_{h\to 0}
\int_{\Omega^1} \chi_h(x)\mathcal{Q}_3(x')\Big(G^h(x)\Big)~\mbox{d}x  \\
& = \frac{1}{2} \liminf_{h\to 0}\int_{\Omega^1} 
\mathcal{Q}_3(x')\Big(\chi_h(x)G^h(x)\Big)~\mbox{d}x  \\
& \geq \frac{1}{2}\int_{\Omega^1} \mathcal{Q}_3(x')\Big(G(x)\Big)~\mbox{d}x
\geq \frac{1}{2}\int_{\Omega^1} \mathcal{Q}_2(x')\Big(G(x)_{tan}\Big)~\mbox{d}x.
\end{split}
\end{equation*}
Above, we used the fact that $\chi_hG^h$ converges weakly in 
$L^2(\Omega^1,\mathbb{R}^{3\times 3})$ to $G$ (compare with the convergence 
in (\ref{pr4.5}))
and the nonnegative definiteness of the quadratic forms 
$\mathcal{Q}_3(x')$, following from $A(x')$ being the minimizer 
of the mapping $W$, as above.

By (\ref{pr5}):
\begin{equation*}
\begin{split}
\mathcal{Q}_2(x')\Big(G(x',x_3)_{tan}\Big) = 
\mathcal{Q}_2(x')\Big(G(x',0)_{tan}\Big) & + 2 x_3 \mathcal{L}_2(x')
\Big( G(x',0)_{tan}, [AQ^{-1}\nabla \vec n]_{tan}(x')\Big)\\
& + x_3^2\mathcal{Q}_2(x')\Big([AQ^{-1}\nabla \vec n]_{tan}(x')\Big).
\end{split}
\end{equation*}
The second term above, expressed in terms of the bilinear operator 
$\mathcal{L}_2(x')$ representing the quadratic form $\mathcal{Q}_2(x')$, 
integrates to $0$ on the domain $\Omega^1$ symmetric in $x_3$.
After dropping the first nonnegative term, we arrive at:
\begin{equation*}
\liminf_{h\to 0}\frac{1}{h^2} I^h(u^h) \geq \frac{1}{24} \int_\Omega 
\mathcal{Q}_2(x') \Big([AQ^{-1}\nabla \vec n]_{tan}(x')\Big)~\mbox{d}x'.
\end{equation*}
This yields the formula in (iii), as $AQ^{-1} \in SO(3)$ so
$AQ^{-1} = (QA^{-1})^T = A^{-1}Q^T$. The proof of Theorem \ref{liminf}
is now complete.
\endproof

\section{Two lemmas on the quadratic forms $\mathcal{Q}_3$
and $\mathcal{Q}_2$}

We now gather some facts regarding the quadratic forms $\mathcal{Q}_3$
and $\mathcal{Q}_2$. First, it is easy to notice that the tangent space
to $\mathcal{F}(x)$ at $A(x)$ coincides with the products of the skew-symmetric
matrices with $A(x)$, denoted by $\mbox{skew}\cdot A(x)$.
Here '$\mbox{skew}$' stands for the space of all skew matrices
of appropriate dimension; in the present case $3\times 3$.
The orthogonal complement of this space equals to:
$$\Big(T_{A(x)}\mathcal{F}(x)\Big)^{\perp} = \mbox{sym}\cdot A(x)^{-1},$$
where '$\mbox{sym}$' denotes the space of all symmetric matrices 
(again, of appropriate dimension).
The quadratic, nonnegative definite form $\mathcal{Q}_3(x')$ is strictly
positive definite on the space above, and it
depends only on the projection of its argument on this space:
\begin{equation}\label{r1}
\mathcal{Q}_3(x')(F) 
= \mathcal{Q}_3(x')\left(\mathbb{P}_{\{\mbox{sym}\cdot A(x')^{-1}\}} F\right).
\end{equation}
\begin{lemma}\label{lem1}
We have:
$$ \mathbb{P}_{\{\mathrm{sym}\cdot A^{-1}\}} 
\left[\begin{array}{cc} \Big[F_{\alpha\beta}\Big]
& \begin{array}{c} f_{13}\\ f_{23}\end{array}\\
\begin{array}{cc} f_{31} & f_{32}\end{array} & f_{33} \end{array}\right]
= \left[\begin{array}{cc} 
\Big[\mathbb{P}_{\{\mathrm{sym}\cdot A_{\alpha\beta}^{-1}\}} F_{\alpha\beta}\Big]
& \begin{array}{c} b_{1}\\ b_{2}\end{array}\\
\left[\begin{array}{cc} b_{1} &  b_{2} \end{array}\right] A_{\alpha\beta}^{-1}
& b_{3} \end{array}\right]$$
with $b_{3} = -f_{33}$ and:
$\left[\begin{array}{cc}b_{1}& b_{2}\end{array}\right] 
\left(\mathrm{Id} + A_{\alpha\beta}^{-2}\right)
= - \left[\begin{array}{cc}f_{13}& f_{23}\end{array}\right] 
- \left[\begin{array}{cc}f_{31}& f_{32}\end{array}\right] 
A_{\alpha\beta}^{-1} $.
\end{lemma}
\begin{proof}
Since the projection $\mathbb{P}$ is a linear operator, we will separately prove
the above formula in two cases: when $F_{\alpha\beta}=0$ and when $f_{ij}=0$.
Notice first that $\mathbb{P}_{\{\mathrm{sym}\cdot A^{-1}\}} F= BA^{-1}$, for a 
symmetric matrix $B$, uniquely determined through the formula:
$$\forall S\in\mbox{sym} \qquad 0 = (F-BA^{-1}): (SA^{-1}).$$
Since the right hand side above equals to $(FA^{-1} - BA^{-2}):S$,
we obtain:
\begin{equation}\label{a1}
FA^{-1} - BA^{-2} \in\mbox{skew}.
\end{equation}
Also, we notice the form of the matrix:
\begin{equation}\label{a2} 
B = \left[\begin{array}{cc} \Big[B_{\alpha\beta}\Big]
& \begin{array}{c} b_{1}\\ b_{2}\end{array}\\
\begin{array}{cc} b_{1} & b_{2}\end{array} & b_{3} \end{array}\right],
\qquad 
BA^{-1} = \left[\begin{array}{cc} 
\Big[B_{\alpha\beta} A_{\alpha\beta}^{-1}\Big]
& \begin{array}{c} b_{1}\\ b_{2}\end{array}\\
\left[\begin{array}{cc} b_{1} &  b_{2} \end{array}\right] A_{\alpha\beta}^{-1}
& b_{3} \end{array}\right]
\end{equation}

\medskip

In the first case when $f_{ij}=0$, let $B_{\alpha\beta} 
= \big[\mathbb{P}_{\{\mathrm{sym}\cdot A^{-1}\}} F_{\alpha\beta}\big] A_{\alpha\beta}$.
Then $F_{\alpha\beta} A_{\alpha\beta}^{-1} - B_{\alpha\beta}A_{\alpha\beta}^{-2}\in
\mbox{skew}$, and the same matrix provides the only non-zero, principal
$2\times 2$ minor of the $3\times 3$ matrix  
$F A^{-1} - BA^{-2}$, where $B$ is taken so that all $b_i=0$ and 
$B_{tan}=B_{\alpha\beta}$. By uniqueness of the symmetric matrix $B$ 
satisfying (\ref{a1}), this proves the claim.

\medskip

In the second case when $F_{\alpha\beta}=0$, define $B$ as in 
(\ref{a2}) with $B_{\alpha\beta}=0$. The result follows, since:
$$FA^{-1} - BA^{-2} = \left[\begin{array}{cc} 
\Big[ 0 \Big]
& \begin{array}{c} f_{13}\\ f_{23}\end{array}\\
\left[\begin{array}{cc} f_{31} &  f_{32} \end{array}\right] A_{\alpha\beta}^{-1}
& f_{33} \end{array}\right] +
\left[\begin{array}{cc} 
\Big[ 0 \Big]
& \begin{array}{c} b_{1}\\ b_{2}\end{array}\\
\left[\begin{array}{cc} b_{1} &  b_{2} \end{array}\right] A_{\alpha\beta}^{-2}
& b_{3} \end{array}\right],$$
and (\ref{a1}) is equivalent to the conditions on $b_i$ given in the statement 
of the lemma. We remark that since $A_{\alpha\beta}^{-2}= [g_{\alpha\beta}]^{-1}$
is strictly positive definite, then the same is true for the
matrix $\mbox{Id} + A_{\alpha\beta}^{-2}$, which implies its invertibility.
\end{proof}

Now, the quadratic and nonnegative definite form $\mathcal{Q}_2(x')$
is likewise strictly positive definite on the space
$\mbox{sym}\cdot A_{\alpha\beta}^{-1}$ and:
\begin{equation}\label{r2}
\mathcal{Q}_2(x')(F_{tan}) 
= \mathcal{Q}_2(x')\left(\mathbb{P}_{\{\mbox{sym}\cdot A(x')^{-1}\}} F_{tan}\right).
\end{equation}

\begin{lemma}\label{lem2}
There exists linear maps $b,c:\mathbb{R}^{2\times 2}\longrightarrow \mathbb{R}^3$
related by:
$$ c\left(F_{tan}\right) = -  
\left[\begin{array}{cc}\Big[\mathrm{Id} + A_{\alpha\beta}^{-2} \Big]
& \begin{array}{c} 0\\ 0\end{array}\\
\begin{array}{cc} 0 &  0 \end{array}
& 1 \end{array}\right]\cdot b\left(F_{tan}\right)$$
and such that:
\begin{equation}\label{rep}
\mathcal{Q}_2(x')(F_{tan}) = 
\mathcal{Q}_3(x')\left[\begin{array}{cc} 
\Big[\mathbb{P}_{\{\mathrm{sym}\cdot A_{\alpha\beta}^{-1}\}} F_{tan}\Big]
& \begin{array}{c} b_{1}\\ b_{2}\end{array}\\
\left[\begin{array}{cc} b_{1} &  b_{2} \end{array}\right] A_{\alpha\beta}^{-1}
& b_{3} \end{array}\right]
= \mathcal{Q}_3(x')\left[\begin{array}{cc} \Big[F_{tan}\Big]
& \begin{array}{c} c_{1}\\ c_{2}\end{array}\\
\begin{array}{cc} 0 &  0 \end{array} & c_{3} \end{array}\right].
\end{equation}
\end{lemma}
\begin{proof}
By (\ref{r1}), Lemma \ref{lem1} and the definition of $\mathcal{Q}_2$ 
it follows that:
$$\mathcal{Q}_2(x')(F_{tan}) = \min_{b\in\mathbb{R}^3}
\mathcal{Q}_3(x')\left[\begin{array}{cc} 
\Big[\mathbb{P}_{\{\mathrm{sym}\cdot A_{\alpha\beta}^{-1}\}} F_{tan}\Big]
& \begin{array}{c} b_{1}\\ b_{2}\end{array}\\
\left[\begin{array}{cc} b_{1} &  b_{2} \end{array}\right] A_{\alpha\beta}^{-1}
& b_{3} \end{array}\right].$$
Hence we obtain the first equality in the representation (\ref{rep}).
The second one follows again from (\ref{r1}) and Lemma \ref{lem1}, provided that
$c_{3} = -b_{3}$ and
$\left[\begin{array}{cc}c_{1}& c_{2}\end{array}\right]
= - \left[\begin{array}{cc}b_{1}& b_{2}\end{array}\right] 
\left(\mathrm{Id} + A_{\alpha\beta}^{-2}\right)$,
which is exactly the condition defining the vector $c$
in the statement of the lemma.
\end{proof}

\section{The recovery sequence - a proof of Theorem \ref{threcseq}}

Following the reasoning in step 1 of the proof of Theorem \ref{liminf}, 
we first notice that the matrix field $Q$ whose columns are given by:
$$Q(x') = \Big[\partial_1 y(x'), \partial_2 y(x'), \vec n(x')\Big]
\in \mathcal{F}(x').$$
Hence in particular: $QA^{-1}\in SO(3)$.
With the help of the above definition and Lemma \ref{lem2}, we put:  
\begin{equation}\label{ddef}
d(x') = Q(x')A^{-1}(x')\cdot c\Big(A_{\alpha\beta}^{-1}(\nabla y)^T
\nabla \vec n(x')\Big)\in L^2(\Omega, \mathbb{R}^3).
\end{equation}
Let $d^h\in W^{1,\infty}(\Omega,\mathbb{R}^3)$ be such that:
\begin{equation}\label{approxd}
\lim_{h\to 0} d^h = d \quad \mbox{in } L^2(\Omega) \qquad
\mbox{ and } \qquad
\lim_{h\to 0} h\| d^h\|_{W^{1,\infty}}  =0.
\end{equation} 
Note that a sequence $d^h$ with properties (\ref{approxd}) can always be 
derived by reparametrizing (slowing down) a sequence of smooth approximations 
of the given vector field $d\in L^2(\Omega)$.

Recalling (\ref{normal}), we now approximate $y$ and $\vec n$ respectively by sequences 
$y^h\in W^{2,\infty}(\Omega,\mathbb{R}^3)$  and 
$\vec n^h \in W^{1,\infty} (\Omega, {\mathbb R}^3)$ such that:
\begin{equation}\label{yhapprox}
\begin{split}
& \lim_{h\to 0} \|y^h - y\|_{W^{2,2}(\Omega)} = 0, \qquad 
\lim_{h\to 0} \|\vec n^h - \vec n\|_{W^{1,2}(\Omega)} =0, \\ & 
h\Big ( \|y^h\|_{W^{2,\infty}(\Omega)} + \|\vec n^h\|_{W^{1,\infty}(\Omega)}\Big ) 
\leq \varepsilon_0,\\
&  \lim_{h\to 0}\frac{1}{h^2} \Big|\left\{x'\in \Omega; ~ y^h(x') \neq y(x')\right\} 
\cup \left\{x'\in \Omega; ~ \vec n^h(x') \neq \vec n(x')\right\}\Big| =0,
\end{split}
\end{equation}  
for a sufficiently small, fixed number $\varepsilon_0>0$, to be chosen later.
The existence of such approximation follows by partition of unity and a truncation
argument, as a special case of the Lusin-type result for Sobolev functions 
in \cite{Liu50} (see also Proposition 2 in \cite{FJMhier}).  

Define:
\begin{equation}\label{recseq}
u^h(x', x_3) = y^h(x') + x_3 \vec n^h(x') + \frac{x_3^2}{2} d^h(x').
\end{equation}
Note that each map: 
$\Omega\ni x'\mapsto \mbox{dist}(\nabla u^h(x'), \mathcal F(x'))$ 
vanishes on $\Omega_h$ and is 
Lipschitz in $\Omega$, with Lipschitz constant of order $O(1/h)$.
Here, we let:
$$\Omega_h= \left\{x'\in \Omega; ~ y^h(x') = y(x') \mbox{ and } 
\vec n^h(x') =\vec n(x') \right\}.$$ 
For any point $x'\in \Omega \setminus \Omega_h$, 
we also have $\mbox{dist}^2 (x', \Omega_h) \leq C |\Omega \setminus \Omega_h|$. 
The proof of the latter statement is standard, see for example \cite{lemopa1}, 
Lemma 6.1 for a similar argument. As a consequence, by (\ref{yhapprox}) we obtain 
$1/h^2 \mbox{dist}^2 (x', \Omega_h) \to 0$ and hence: 
\begin{equation}\label{bdd-distance}
\mbox{dist}(\nabla u^h(x'), \mathcal F(x')) \leq O(1/h) \mbox{dist}(x', \Omega_h)  
= o(1).
\end{equation}  
The gradient of the deformation $u^h$ is given by:
$$\nabla u^h(x',x_3) = Q^h(x') + x_3 A_2^h(x') + \frac {x_3^2}{2} D^h(x'), $$
where:  
\begin{equation*}
\begin{split}
&  Q^h(x') = Q(x') \quad \mbox{ in } \Omega_h, \qquad 
A_2^h(x') = \Big[\partial_1 \vec n^h(x'),\partial_2 \vec n^h(x'), d^h(x')\Big], \\ 
& \lim_{h\to 0} A_2^h = A_2 =  
\Big[\partial_1\vec n, \partial_2\vec n, d\Big] \quad \mbox{ in } L^2(\Omega), \\
& D^h = \Big [\partial_1 d^h, \partial_2 d^h, 0 \Big ]. 
\end{split} 
\end{equation*}  
Note that by (\ref{bdd-distance}) and the local $\mathcal{C}^2$ regularity of $W$, 
the quantity $W(x,\nabla u^h(x))$ remains bounded upon
choosing $h$ and $\varepsilon_0$ in (\ref{yhapprox}) small enough. 
The convergence in (i) Theorem \ref{liminf} follows immediately. 

We now prove (\ref{limexact}). 
Using  Taylor's expansion of $W$ in a neighborhood of $Q(x')$, we obtain:  
\begin{equation*}
\begin{split}
\frac{1}{h^2} I^h(u^h) &= \frac{1}{h^2}\int_{\Omega^1_h}
W\left(x, Q(x') + hx_3 A_2^h(x') +   h^2 \frac{x_3^2}{2} D^h(x')\right)~\mbox{d}x  
+ \frac{1}{h^2}\int_{\Omega^1 \setminus \Omega^1_h}
W(x, \nabla u^h(x)) ~\mbox{d}x \\
& = \int_{\Omega^1_h} \left(\frac{1}{2}\nabla^2 W (x', \cdot)_{\mid Q(x')} 
(x_3 A_2^h(x'), x_3A_2^h(x')) +  \mathcal{R}^h(x)\right) ~\mbox{d}x 
+ \frac{O(1)}{h^2} |\Omega \setminus \Omega_h|.
\end{split}
\end{equation*} 
Here the reminder $\mathcal{R}^h$ converges, by (\ref{approxd}), to $0$ pointwise 
almost everywhere, as $h\to 0$.
Therefore, recalling the boundedness of $W(x,\nabla u^h(x))$ 
we deduce by dominated convergence and (\ref{yhapprox}) that the above 
integral converges, as $h\to 0$, to:
\begin{equation*}
\begin{split}
\frac{1}{2}\int_{\Omega^1} x_3^2&\nabla^2 W(x', \cdot)_{\mid Q(x')} (A_2(x'), A_2(x')) 
~\mbox{d}x = \frac{1}{2}\int_{\Omega^1} x_3^2 \mathcal{Q}_3(x') \Big(AQ^{-1}A_2\Big) 
~\mbox{d}x\\
& = \frac{1}{24}\int_{\Omega}  \mathcal{Q}_3(x') \Big(A^{-1}Q^{T}A_2\Big) ~\mbox{d}x' 
= \frac{1}{24}\int_{\Omega}\mathcal{Q}_2\Big(A_{\alpha\beta}^{-1} (\nabla y)^T 
\nabla\vec n\Big)~\mbox{d}x'. 
\end{split}
\end{equation*}
where we applied frame invariance,  (\ref{ddef}) and (\ref{rep}).
\endproof

\section{Conditions for existence of  $W^{2,2}$ isometric immersions 
of Riemannian metrics - a proof of Theorem \ref{isoim}}\label{isometrie}

The assertions  in (i) follow directly from Theorem \ref{liminf}
and Theorem \ref{threcseq}. It remains to prove (ii), which clearly 
implies (iii).

Assume that $\lim_{h\to 0}\frac{1}{h^2}I^h(u^h) = 0$ 
for some sequence of deformations
$u^h\in W^{1,2}(\Omega^h,\mathbb{R}^3)$. Then, by Theorem \ref{liminf}
there exists a metric realization $y\in W^{2,2}(\Omega,\mathbb{R}^3)$
such that:
$$ \int_\Omega \mathcal{Q}_2(x') \Big(A_{\alpha\beta}^{-1} \Pi(x')\Big) 
~\mbox{d}x' = 0,$$
where $\Pi = (\nabla y)^T\nabla \vec n$ is the second fundamental form of
the image surface $y(\Omega)$.
Recalling (\ref{r2}) we obtain:
$$0 = \mathcal{Q}_2(x') \Big(A_{\alpha\beta}^{-1}\Pi\Big) = 
\mathcal{Q}_2(x') \Big(\mathbb{P}_{\mbox{sym}\cdot A_{\alpha\beta}^{-1}} 
(A_{\alpha\beta}^{-1}\Pi)\Big)\qquad \forall x'\in\Omega.$$
Since the quadratic form $\mathcal{Q}_2(x')$ is nondegenerate 
on $\mbox{sym}\cdot A_{\alpha\beta}^{-1}$, it follows that:
\begin{equation}\label{B} 
B A_{\alpha\beta}^{-1} = \mathbb{P}_{\mbox{sym}\cdot A_{\alpha\beta}^{-1}} 
(A_{\alpha\beta}^{-1}\Pi) =0, 
\end{equation}
for the symmetric matrix $B\in\mathbb{R}^{2\times 2}$ satisfying:
$$ (A_{\alpha\beta}^{-1}\Pi - BA_{\alpha\beta}^{-1}): (SA_{\alpha\beta}^{-1})
= 0 \qquad \forall S\in\mbox{sym}. $$
The above condition is equivalent to $A_{\alpha\beta}^{-1}\Pi A_{\alpha\beta}^{-1} 
- BA_{\alpha\beta}^{-2}\in\mbox{skew}$, but $B=0$ in view of (\ref{B}),
so:
$$A_{\alpha\beta}^{-1} \Pi A_{\alpha\beta}^{-1} \in\mbox{skew}.$$
Since $\Pi\in\mbox{sym}$, there must be $\Pi=0$ and therefore indeed effectively
$y:\Omega\longrightarrow \mathbb{R}^2$.

On the other hand, if $y\in W^{2,2}(\Omega, \mathbb{R}^2)$ is a $2$d realization 
of $[g_{\alpha\beta}]$ then clearly $\Pi = (\nabla y)^T\nabla \vec n = 0$,
so for the recovery sequence corresponding to $y$ and constructed in Theorem
\ref{threcseq}, there holds $\lim_{h\to 0}\frac{1}{h^2} I^h(u^h) = I(y) = 0$.
\endproof

\end{document}